\documentclass[11pt]{article}

\usepackage[utf8]{inputenc}
\usepackage[british]{babel}
\usepackage[a4paper, margin=2.5cm]{geometry}
\usepackage{amsmath, amsthm, amssymb, amsfonts}
\usepackage{mathabx}
\usepackage{mathtools}
\usepackage{thmtools}
\usepackage[shortlabels]{enumitem}
\usepackage[backref]{hyperref}
\usepackage[font={small, it}]{caption}
\usepackage{microtype}
\usepackage{xcolor}
\usepackage{bm}
\usepackage{tcolorbox}
\usepackage{bbm}

\newcommand{\eps}{\varepsilon}
\renewcommand{\phi}{\varphi}
\renewcommand{\rho}{\varrho}

\DeclareMathOperator*{\E}{\mathbb{E}}

\DeclareMathOperator*{\R}{\mathbb{R}}

\DeclarePairedDelimiter\ceil{\lceil}{\rceil}

\let\setminus=\smallsetminus

\declaretheorem[parent=section]{theorem}

\declaretheorem[sibling=theorem]{proposition}
\declaretheorem[sibling=theorem]{claim}

\setlist{itemsep=0.1em, topsep=0.1em, parsep=0.1em, partopsep=0.1em}

\hypersetup{
  colorlinks,
  linkcolor={red!60!black},
  citecolor={green!50!black},
  urlcolor={blue!80!black}
}

\colorlet{RoyalRed}{red!70!black}
\definecolor{RoyalBlue}{rgb}{0.25, 0.41, 0.88}
\definecolor{RoyalAzure}{rgb}{0.0, 0.22, 0.66}

\newcommand{\srbsquare}{\rotatebox{45}{\tiny{\ensuremath{\blacksquare}}}}
\newcommand{\asscheck}{\hskip.2em \null \hfill \srbsquare}

  {\small
    \begin{adjustwidth*}{1em}{1em}
      {\bf Verifying the assumptions of #1.}%
  }%
  {\asscheck\end{adjustwidth*}}

\title{An Optimal Decentralized \texorpdfstring{$(\Delta + 1)$}{(Delta +
1)}-Coloring Algorithm}

\author{
  Daniel Bertschinger\thanks{Email: \{\texttt{daniel.bertschinger}\textbar
      \texttt{johannes.lengler}\textbar \texttt{anders.martinsson}\textbar
      \texttt{steger}\textbar \texttt{mtrujic}\textbar
      \texttt{emo}\}\texttt{@inf.ethz.ch}}
  \and
  Johannes Lengler\footnotemark[1]
  \and
  Anders Martinsson\footnotemark[1]
  \and
  Robert Meier\thanks{Email: \texttt{romeier@student.ethz.ch}}
  \and
  Angelika Steger\footnotemark[1]
  \and
  Milo\v{s} Truji\'{c}\footnotemark[1] \textsuperscript{,}\thanks{author was
  supported by grant no.\ 200021 169242 of the Swiss National Science
  Foundation.}
  \and
  Emo Welzl\footnotemark[1]
}

\makeatletter
\renewcommand\@date{{%
    Department of Computer Science \\ ETH Z\"{u}rich, Switzerland
    \date{}
}}
\makeatother

\begin{document}

\maketitle

\begin{abstract}
  Consider the following simple coloring algorithm for a graph on $n$ vertices.
  Each vertex chooses a color from $\{1, \dotsc, \Delta(G) + 1\}$ uniformly at
  random. While there exists a conflicted vertex choose one such vertex
  uniformly at random and recolor it with a randomly chosen color. This
  algorithm was introduced by Bhartia et al. [MOBIHOC'16] for channel selection
  in WIFI-networks. We show that this algorithm always converges to a proper
  coloring in expected $O(n \log \Delta)$ steps, which is optimal and proves a
  conjecture of Chakrabarty and de Supinski~[SOSA'20].
\end{abstract}

\section{Introduction}

It is well known that an undirected graph $G=(V,E)$ with maximum degree $\Delta
= \Delta(G)$ can be properly colored by using $\Delta+1$ colors. In fact, a
simple greedy algorithm which assigns the colors successively achieves this
bound by just touching each vertex once. Note that the bound $\Delta+1$ is
tight, as cliques and odd cycles require this number of colors.

In \cite{bhartia2016iq} Bhartia et al.\ introduced the use of a simple
decentralized coloring algorithm as an efficient solution to the channel
selection problem in wireless networks. Their algorithm can be formulated as
follows.

\begin{tcolorbox}[colback=black!15!white, title=\scshape{Decentralized Graph
Coloring}]
  \label{alg:decentr-coloring}

  For a graph $G = (V, E)$
  \begin{enumerate}
    \item choose for each vertex $v \in V$ a color from $\{1, \dotsc, \Delta +
      1\}$ independently and uniformly at random;
    \item\label{alg-vertex} choose a vertex $v \in V$ uniformly at random among
      all vertices which have a neighbor in the same color;
    \item\label{alg-colour} recolor $v$ into a color chosen from $\{1, \dotsc,
      \Delta + 1\}$ uniformly at random;
    \item repeat steps \ref{alg-vertex} and \ref{alg-colour} until a proper
      coloring of $G$ is found.
  \end{enumerate}
\end{tcolorbox}

They showed that this algorithm finds a proper coloring in $O(n\Delta)$ rounds
in expectation. Chakrabarty and de Supinski~\cite{chakrabarty2020decentralized}
introduced a variant of the coloring algorithm: instead of recoloring a vertex
$v$ once as above, in their `Persistent Decentralized Coloring Algorithm' such a vertex
$v$ persistently (hence the name) gets recolored until it has no neighbor in the
same color. They showed that this modified algorithm only requires
$O(n\log\Delta)$ recolorings and conjectured that the same bound also holds for
the original algorithm. In this paper we prove their conjecture.

\begin{theorem}\label{thm:main}
  The decentralized coloring algorithm converges in expectation to a proper
  $(\Delta+1)$-coloring in $O(n\log\Delta)$ recoloring steps.
\end{theorem}

In fact, our argument shows that the same runtime bound holds true if the
initial coloring is chosen adversarially. This is in contrast with the
persistent version of the algorithm mentioned above, as that one takes
$\Theta(n\Delta)$ recolorings in expectation when starting with an adversarial
coloring (see~\cite[Theorem~3]{bhartia2016iq}). However, the question raised in
\cite{chakrabarty2020decentralized} of `which algorithm is faster in the random
setting' remains open.

We note that the bound in Theorem~\ref{thm:main} is best possible, as for the
complete graph $K_n$ the decentralized coloring algorithm essentially performs a
{\em Coupon Collector} process. Indeed, once a color (coupon) has been acquired
it remains in the graph until the end of the process and we need to see all
colors. The claim thus follows from the well known fact that in expectation the
coupon collector process with $n$ coupons requires $nH_n = \Theta(n \log n)$
rounds, where $H_k = \sum_{i = 1}^{k} \frac{1}{i}$ is the $k$-th Harmonic
number. Moreover, the result is tight for every combination of $n$ and $\Delta$.
Namely, consider a vertex-disjoint union of $n/\Delta$ complete graphs
$K_\Delta$ and by the same argument the process requires $\Theta(n/\Delta \cdot
\Delta\log\Delta) = \Theta(n\log\Delta)$ rounds.

Our proof of Theorem~\ref{thm:main} is short and elegant, and is based on
\emph{drift analysis}~\cite{Lengler2020}. It is presented in an expository way
and provides insight in why our potential function is appropriate for the
analysis. We complement the analyis by tail bounds in Section 3.

Finally, we conclude the paper by a brief discussion of the parallelized version of this
algorithm, where all `conflicted' vertices get recolored simultaneously (instead
of Step~\ref{alg-colour}), and we prove that this variant takes exponential
time.

\section{Proof of Theorem~\ref{thm:main}}

We start with introducing some notation. We use $c_t$ to denote the coloring of
the graph after $t$ recoloring steps, that is $c_t$ is a function $c_t \colon
V(G) \to \{1, \dotsc, \Delta + 1\}$. With $M_t \subseteq E$ we denote the set of
{\em monochromatic edges} in $c_t$. Observe that $c_t$ is a proper coloring of
$G$ if and only if $|M_t| = 0$. Our main goal is thus to establish good bounds
on (the reduction of) the size of the sets $M_t$. In order to do so it is
helpful to view the recoloring step(s) (i.e.\ Step~\ref{alg-vertex} and
Step~\ref{alg-colour}) as a (slightly different) three step process:

\begin{tcolorbox}[colback=black!15!white, title=\scshape{Recoloring Step}]
  For every $t \geq 1$,
  \begin{enumerate}[\hspace{0.2em} S1]
    \item\label{proc-comp} choose a monochromatic connected component $C
      \subseteq M_{t-1}$ at random proportional to the number of vertices in
      $C$;
    \item\label{proc-select-vertex} choose a vertex $v \in V(C)$ uniformly at
      random;
    \item\label{proc-recolour} let $c_t(v)$ be a uniformly at random chosen
      color from $\{1, \dotsc, \Delta + 1\}$ and set $c_t(u) = c_{t - 1}(u)$ for
      all $ u \neq v$.
  \end{enumerate}
\end{tcolorbox}

As a main tool in bounding the expected number of recoloring steps we use a
so-called drift theorem (see~\cite[Theorem~2.3.1]{Lengler2020}).

\begin{theorem}[Additive Drift Theorem~\cite{he2004study}]
  \label{drift:simple}
  Let $(X_t)_{t \geq 0}$ be a sequence of non-negative random variables with a
  finite state space $\mathcal{S} \subset \R^+_0$ such that $0 \in \mathcal{S}$.
  Let $T := \inf \{ t \geq 0 \mid X_t = 0 \}$. If there exists $\delta > 0$ such
  that for all $s \in \mathcal{S} \setminus \{ 0 \}$ and for all $t > 0$,
  \begin{equation}\label{eq:drift}
    \E[X_{t} - X_{t-1} \mid X_t = s] \leq - \delta,
  \end{equation}
  then
  \[
    \E[T] \leq {\E[X_0]}\;/\;{\delta}.
  \]
\end{theorem}

By {\em drift} we refer to the expectation $\E[X_t - X_{t - 1} \mid X_t = s]$.
(For a more extensive introduction to drift analyis, we refer the reader to~\cite{Lengler2020}.) Our goal is to apply
Theorem~\ref{drift:simple} by assigning to each coloring $c_t$ a real value
$\Phi(t)$ (which we plug in for $X_t$) so that $\Phi(t) = 0$ if and only if
$c_t$ is a proper coloring. The {\em potential function} $\Phi(\cdot)$ we
eventually use to prove Theorem~\ref{thm:main} consists of several terms (see
equation~\eqref{eq:phi} below) and in order to motivate each of the terms we
introduce them one by one. The simplest and most natural choice is to consider
just the number of monochromatic edges, i.e.\ $\Phi(t) := |M_t|$. (Mind that
this is only for explanatory purposes and will not be the final definition of
$\Phi$.) To apply Theorem~\ref{drift:simple} we need to estimate the drift from
a single recoloring step. In the following claim (and in fact all similar ones
in this section), the expectation is always taken with respect to a single
recoloring step. That is, we (implicitly) condition on the coloring~$c_{t - 1}$
without stating it every time. Note that this formulation implies what is
required by equation~\eqref{eq:drift}.

As it turns out, in the case of $\Phi(t) := |M_t|$ we do not need to make use of
the fact that the component $C$ is chosen randomly, we may assume that the
component $C$ is given arbitrarily or even by an adversary.

\begin{claim}
  \label{cl:component-edges}
  For all $t \geq 1$ and any connected component $C$ in $M_{t - 1}$ we have
  \[
    \E\big[ |M_t| \bigm\vert C \big] \leq |M_{t - 1}| - \bar d(C) + 1 -
    \frac{1}{\Delta+1},
  \]
  where $\bar d(C)$ denotes the average degree of the graph induced by $V(C)$.
\end{claim}
\begin{proof}
  The claim follows easily from the following two observations. As $v$ is chosen
  uniformly at random within $C$ (as in Step~\ref{proc-select-vertex}), we
  decrease the number of monochromatic edges within $C$ by $\bar{d}(C)$ whenever
  the newly chosen color is different from the current color of $C$, which
  happens with probability ${\Delta}/{(\Delta +1)}$. All edges incident to $v$
  that do not belong to $C$ become monochromatic with probability $1/{(\Delta
  +1)}$. Thus we have
  \[
    \E\big[ |M_t| \bigm\vert C \big] \leq |M_{t-1}| - \bar d(C) \cdot
    \frac{\Delta}{\Delta+1} + \frac{\Delta - \bar d(C)} {\Delta+1} = |M_{t - 1}|
    - \bar d(C) + 1 - \frac{1}{\Delta+1},
  \]
  as claimed.
\end{proof}

As the average degree of every monochromatic component in $M_t$ is at least one,
Claim~\ref{cl:component-edges} implies $\E[|M_t|] \le |M_{t- 1}| - 1/(\Delta+1)$
whenever $|M_{t-1}|>0$. The following proposition then easily follows from
Theorem~\ref{drift:simple}.

\begin{proposition}\label{prop:endgame}
  Let $D > 0$ be any fixed constant. For every graph $G$ and every coloring
  $c_0$ of $G$ such that $|M_0| \le Dn / \Delta$ the decentralized coloring
  algorithm reaches a proper $(\Delta + 1)$-coloring in expectation after $O(n)$
  recoloring steps.
\end{proposition}

Unfortunately, a random coloring of a graph $G$ with $\Delta + 1$ colors has in
expectation $\Theta(n)$ monochromatic edges, so Proposition~\ref{prop:endgame}
is not immediately applicable. Instead, Claim~\ref{cl:component-edges} together
with Theorem~\ref{drift:simple} only provide us with the bound of $O(n\Delta)$
(see Bhartia et al.~\cite{bhartia2016iq}). In order to go beyond this, observe that
Claim~\ref{cl:component-edges} actually gives a drift of $-1/3$ whenever
$|V(C)|\ge 3$, as the average degree of a connected graph on $s\ge 3$ vertices
is at least $4/3$. Thus, the only critical case are components $C$ that consist
of only one edge. To handle these we introduce some more notation.

We denote by $I_t \subseteq M_t$ the set of {\em isolated edges}, that is all
edges which are monochromatic components of size two. We also let $P_t \subseteq
V$ stand for the set of all properly colored vertices, i.e.\ the vertices that
are not incident to any edge in $M_t$. Akin to Claim~\ref{cl:component-edges},
the next claim gives a bound on the expected change in the number of isolated
edges in one recoloring step.

\begin{claim}
  \label{cl:component-isolated-edges}
  For all $t \geq 1$ and any connected component $C$ in $M_{t-1}$ we have
  \[
    \E\big[ |I_t| \bigm\vert C \big] \leq |I_{t - 1}| + \bar d(C) + 1.
  \]
  For components $C$ that form an isolated edge, we have in addition
  \[
    \E\big[ |I_t| \bigm\vert C = uw \big] \leq |I_{t - 1}| -
    \frac{\Delta}{\Delta+1} + \frac{|N(u)\cap P_{t-1}| + |N(w)\cap
    P_{t-1}|}{2(\Delta+1)}.
  \]
 \end{claim}
\begin{proof}
  By recoloring a vertex $v$, the only isolated edges that can be created are
  edges that are incident to neighbors of $v$ within $C$ (at most one isolated
  edge per neighbor of $v$) and edges between $v$ and $P_{t-1}$ (naturally at
  most one edge incident to $v$ can be isolated and monochromatic). This proves
  the first inequality. For the second assume that $C=uw$. Clearly, after
  recoloring one of $u$ and $w$ with a different color (which happens with
  probability $\Delta/(\Delta+1)$), the isolated edge $C=uw$ disappears.
  Observe, also that a new isolated edge can only be generated if we choose as a
  new color for $u$ (or $w$) a color of a vertex in $N(u)\cap P_{t-1}$ (or
  $N(w)\cap P_{t-1}$) respectively. This, together with the fact that each of
  $u$ or $w$ is chosen in Step~\ref{proc-select-vertex} with probability $1/2$,
  proves the second inequality.
\end{proof}

We pause for a moment from the proof of Theorem~\ref{thm:main} to showcase the
use of previous claims for proving a positive result about complete bipartite
graphs.

\begin{proposition}\label{prop:knn}
  For complete bipartite graphs $G = K_{n, m}$ the decentralized coloring
  algorithm reaches a proper $(\Delta + 1)$-coloring in expectation after
  $O(\min\{n, m\})$ recoloring steps.
\end{proposition}
\begin{proof}
  Observe that for complete bipartite graphs vertices of $P_{t - 1} \cap A$ and
  $P_{t-1} \cap B$ need to be colored with {\em different} colors (here $A$ and
  $B$ denote the two parts of the bipartite graph). Also note that an isolated
  edge can only be generated if a color appears only once in $P_{t - 1} \cap A$
  (and $P_{t - 1} \cap B$). Therefore, for a monochromatic edge $uv$ we have
  $|N(u) \cap P_{t - 1}| + |N(w) \cap P_{t - 1}| \leq \Delta$. We can thus
  replace the bound in the second inequality of
  Claim~\ref{cl:component-isolated-edges} by
  \begin{equation}\label{eq:bipartite-isolated}
    \E\big[ |I_t| \bigm\vert C = uw \big] \leq |I_{t - 1}| -
    \frac{\Delta}{\Delta+1} + \frac{\Delta}{2(\Delta+1)} \le |I_{t - 1}| -
    \frac{\Delta}{2(\Delta+1)}.
  \end{equation}
  Consider now the potential function $\Phi(t) := |M_t| + \frac{1}{10} |I_t|$.
  In case $|V(C)| \geq 3$, from Claim~\ref{cl:component-edges} and
  Claim~\ref{cl:component-isolated-edges}, we get (with room to spare)
  \begin{align*}
    \E\big[ \Phi(t) \bigm\vert C, |V(C)| \geq 3 \big] & \leq |M_{t - 1}| - \bar
    d(C) + 1 - \frac{1}{\Delta + 1} + \frac{1}{10} |I_{t - 1}| + \frac{1}{10}
    \bar d(C) + \frac{1}{10} \\
    & \leq \Phi(t - 1) - \frac{9}{10} \bar d(C) + \frac{11}{10} \leq \Phi(t -
    1) - \frac{1}{20},
  \end{align*}
  where we used the fact that $\bar d(C) \geq 4/3$. On the other hand, if $C =
  uw$ then by Claim~\ref{cl:component-edges} and \eqref{eq:bipartite-isolated}
  we have
  \begin{align*}
    \E\big[ \Phi(t) \bigm\vert C = uw \big] & \leq |M_{t - 1}| - \bar d(C) + 1 -
    \frac{1}{\Delta + 1} + \frac{1}{10} |I_{t - 1}| - \frac{1}{10}
    \frac{\Delta}{2(\Delta + 1)} \\
    & \leq \Phi(t - 1) - \frac{1}{20} \Big( \frac{20}{\Delta + 1} +
    \frac{\Delta}{\Delta + 1} \Big) \leq \Phi(t - 1) - \frac{1}{20}.
  \end{align*}
  In conclusion,
  \[
    \E\big[ \Phi(t) \bigm\vert C \big] \leq \Phi(t - 1) - \frac{1}{20},
  \]
  for every component~$C$. The proposition now follows from
  Theorem~\ref{drift:simple} together with the fact that in a random
  $x$-coloring an edge is monochromatic with probability $1/x$ and thus
  \[
    \E[\Phi(0)] \leq \E[|M_0|] + \E[|I_0|] \leq 2 \cdot \frac{n \cdot
    m}{\max\{n, m\} + 1} \leq 4 \min\{n, m\},
  \]
  with room to spare.
\end{proof}

We note that this proof actually shows that the assertion of
Proposition~\ref{prop:knn} remains true, for sufficiently small $\eps > 0$, if
we reduce the number of colors to be used by the algorithm to $(1 -
\eps)\Delta$, that is $(1 - \eps)\max\{n, m\}$. We do not elaborate further on
this.

After this short detour we come back to the proof of Theorem~\ref{thm:main}.
What one could conclude from the two claims above is that if we were to choose a
component $C$ in Step~\ref{proc-comp} which is of size at least three throughout
the process, then the drift obtained (Claim~\ref{cl:component-edges}) would
always be less than $-1/3$. However, this is far too optimistic to hope for.

Consider an isolated edge $uv$ and assume we recolor $v$. If the new color
chosen {\em does not} belong to its properly-colored neighborhood $N(v) \cap
P_{t - 1}$, then the number of monochromatic isolated edges decreases by one.
This happens with constant probability unless the size of $N(v) \cap P_{t - 1}$
is close to $\Delta$.

Since in Step~\ref{proc-comp} we choose $C$ {\em randomly}, we expect a strong
drift `towards the target' as long as we are in one of the situations from the
paragraphs above. In other words, we have a desired drift unless $M_{t - 1}$
comprises mostly of isolated edges {\em and} most vertices $u \in V(I_{t - 1})$
have almost $\Delta$ neighbors in $N(u) \cap P_{t - 1}$.

Let us hence analyze what happens if in such a case we recolor a vertex $v$
belonging to an isolated edge $uv$. Suppose we set $c_t(v) := c_{t - 1}(x)$ for
some $x \in N(v) \cap P_{t - 1}$. If the color $c_{t - 1}(x)$ appears multiple
times in $N(v) \cap P_{t - 1}$, we do not create a new isolated edge.
Otherwise, the edge $xv$ becomes isolated and $P_t := (P_{t - 1} \setminus
\{x\}) \cup \{u\}$. However, crucially, as we assumed that every vertex $u \in
V(I_{t - 1})$ had roughly $\Delta$ neighbors in $P_{t - 1}$, we conversely have
that an average vertex in $P_{t - 1}$ has roughly $\Delta |V(I_{t - 1})|/|P_{t -
1}|$ neighbors in $V(I_{t - 1})$. Thus, we may expect that $N(x) \cap P_t$ is
{\em smaller} than $\Delta$. In other words, we expect that $e(V(I_t), P_t)$ is
smaller than $e(V(I_{t - 1}), P_{t - 1})$. Here and throughout we use $e(X, Y)$
to denote the number of edges between two disjoint vertex sets $X$ and $Y$.

Previous considerations motivate keeping track of $e(V(I_t), P_t)$ as well and
lead us to formulate the following potential function:
\begin{equation}\label{eq:phi}
  \Phi(t) := |M_t| + \frac{|I_t|}{10} + \frac{e(V(I_t), P_t)}{100\Delta}.
\end{equation}
Note that the value of $\Phi(t)$ is always proportional to the number of
monochromatic edges.

\begin{claim}\label{cl:mono_phi}
  For all $t\ge 1$ we have
  \[
    |M_t| \le \Phi(t) \le 2|M_t|.
  \]
\end{claim}
\begin{proof}
  The first inequality is trivial. The second follows, with room to spare, as
  $I_t \subseteq M_t$ and $e(V(I_t), P_t) \le 2 |I_t| \cdot \Delta$.
\end{proof}

With Claim~\ref{cl:mono_phi} at hand we deduce from
Proposition~\ref{prop:endgame} that in order to complete the proof of
Theorem~\ref{thm:main} it suffices to show that the algorithm reduces the
potential $\Phi$ to a value of $Dn/\Delta$ in $O(n\log\Delta)$ steps, for some
arbitrarily large but fixed constant~$D > 0$. This is what we do in the
remainder of this section.

Note also that there is no hope to always get a {\em constant} drift, as by
Theorem~\ref{drift:simple} this would then lead to a bound of $O(n)$ recoloring
steps, which would contradict the bound of $\Omega(n\log n)$ for $K_n$. Instead
we show a {\em multiplicative} drift.

\begin{claim}\label{cl:mult}
  For any $t \geq 1$ with $\Phi(t-1) > 0$, we have
  \[
    \E[\Phi(t)] \leq \Phi(t-1) \Big(1 - \frac{1}{1000 n} \Big).
  \]
\end{claim}
\begin{proof}
  By linearity of expectation we can consider each term of $\Phi(\cdot)$ in
  \eqref{eq:phi} independently. The first two terms are handled by
  Claim~\ref{cl:component-edges} and Claim~\ref{cl:component-isolated-edges}, so
  we first establish some bounds on the third. Observe that in order for an edge
  to be counted in $e(V(I_t), P_t)$ but not in $e(V(I_{t-1}), P_{t-1})$ it must
  be incident to a vertex in either $V(I_t)\setminus V(I_{t-1})$ or
  $P_t\setminus P_{t-1}$. Let $C$ be a component chosen in Step~\ref{proc-comp}
  and $v$ a vertex chosen in Step~\ref{proc-select-vertex}. For any vertex in
  $\{v\}\cup (N(v) \cap V(C))$, we either get one new isolated edge or one new
  properly colored vertex (or neither). In the former, the other endpoint of
  that edge potentially contributes by $\Delta$ to $e(V(I_t), P_t)$, and in the
  latter each monochromatic edge with one endpoint in $N(v) \cap V(C)$
  potentially contribute by $\Delta$ to $e(V(I_t), P_t)$ for each of its
  endpoints. Thus we have
  \[
    \E[ e(V(I_t), P_t) \mid C] \le e(V(I_{t-1}), P_{t-1}) + (\bar d(C)+1) \cdot
    2\Delta,
  \]
  where as before $\bar d(C)$ denotes the average degree of the component~$C$.
  Together with Claim~\ref{cl:component-edges} and
  Claim~\ref{cl:component-isolated-edges}, for all components $C$ on at least
  three vertices we get
  \begin{equation}\label{eq:drift:phi:C3}
    \begin{aligned}
      \E\big[ \Phi(t) \bigm\vert C, |V(C)| \ge 3 \big] & \leq \Phi(t - 1) -
      \Big( 1 - \frac1{10} - \frac2{100} \Big) \bar d(C) + 1 +\frac1{10} +
      \frac2{100}
      \\
      &\le \Phi(t - 1) - \frac{1}{25} \bar d(C),
    \end{aligned}
  \end{equation}
  where the last inequality follows from $\bar d(C) \ge 4/3$.

  Next we consider the third term of $\Phi(\cdot)$ conditioned on choosing a
  component $C \subseteq I_{t-1}$, i.e.\ $C$ is an isolated edge. We first let
  $d(v,X) := |N(v) \cap X|$ for all $v\in V$ and sets $X \subseteq V$ and denote
  by
  \[
    \bar d_{IP} := \frac{1}{|V(I_{t - 1})|} \sum_{u \in V(I_{t-1})} d(u,
    P_{t-1}) \qquad \text{and} \qquad \bar d_{PI} := \frac{1}{|P_{t - 1}|}
    \sum_{u \in P_{t - 1}} d(u, V(I_{t - 1}))
  \]
  the average degree of vertices in $V(I_{t-1})$ into $P_{t-1}$, and the average
  degree of vertices in $P_{t - 1}$ into $V(I_{t - 1})$, respectively. Note
  that, of course, we have $\sum_{u \in V(I_{t - 1})} d(u, P_{t - 1}) = \sum_{u \in P_{t
  - 1}} d(u, V(I_{t - 1}))$, and hence $\bar d_{IP}|V(I_{t - 1})| = \bar d_{PI}
  |P_{t - 1}|$.

  Consider an isolated edge $wv$ and assume $v$ gets recolored with a new color.
  Then, since $w$ is now properly colored, all $d(w, P_{t-1})$ edges incident to
  $w$ which contributed to $e(V(I_{t-1}), P_{t-1})$ are not counted in
  $e(V(I_t), P_t)$, except possibly one in case $v$ forms a new isolated edge
  with a neighbor of $w$. Moreover, any new edge counted in $e(V(I_t), P_t)$
  must be incident to either $v$, $w$, or a vertex $x\in P_{t-1}$ for which $vx
  \in I_t$. There are at most $\Delta - d(v, P_{t-1})$, $\Delta- d(w, P_{t-1})$,
  and $\Delta - d(x, V(I_{t-1}))$ such edges respectively not already counted in
  $e(V(I_{t-1}), P_{t-1})$. Combining all this we get
  \begin{multline*}
    e(V(I_t), P_t) \leq e(V(I_{t-1}), P_{t-1}) - d(w, P_{t-1}) + 1 + \Delta-
    d(v, P_{t-1}) \\
    + \Delta- d(w, P_{t-1}) + \sum_{x \in P_{t-1}} \mathbbm{1}_{ vx \in I_t }
    \big(\Delta - d(x, V(I_{t-1})\big),
  \end{multline*}
  if $c_t(v)\neq c_{t-1}(v)$, and of course $e(V(I_t), P_t) = e(V(I_{t-1}),
  P_{t-1})$ if $c_t(v)= c_{t-1}(v)$.

  We conclude that
  \begin{multline*}
    \E\big[ e(V(I_t), P_t) \bigm\vert C \subseteq I_{t-1} \big] \leq e(V(I_{t-1}),
    P_{t-1}) + \frac{\Delta}{\Delta+1}(2\Delta+1-3\bar d_{IP}) \\
    + \frac{1}{|V(I_{t-1})|} \sum_{x\in P_{t-1}} \sum_{v\in V(I_{t-1})\cap N(x)}
    \frac{1}{\Delta+1} \big(\Delta - d(x,V(I_{t-1})) \big),
  \end{multline*}
  where the last term can be rewritten as
  \[
    \frac{1}{|V(I_{t-1})|}\sum_{x\in P_{t-1}} \frac{d(x,V(I_{t-1})) \big(\Delta
    - d(x,V(I_{t-1}))\big)}{\Delta+1}.
  \]
  We note that the summand above can be written as $f(d(x, V(I_{t-1}))$ where
  $f(y):=y(\Delta-y)/(\Delta+1)$ is a concave function. Hence, by Jensen's
  inequality, we can upper bound the expression by $|P_{t-1}|
  f(\bar{d}_{PI})/|V(I_{t-1})|=\bar d_{IP} (\Delta - \bar d_{PI})/(\Delta+1)$.
  Altogether we get
  \begin{align*}
    \E \big[ e(V(I_t), P_t) \bigm\vert C \subseteq I_{t-1} \big] & \leq
    e(V(I_{t-1}), P_{t-1}) + \frac{\Delta(2\Delta + 1 - 3\bar
    d_{IP})}{\Delta+1}  + \frac{ \bar d_{IP} (\Delta - \bar d_{PI})}{\Delta+1} \\
    & \leq e(V(I_{t-1}), P_{t-1}) + \frac{\Delta}{\Delta + 1} (2\Delta + 1 - 2
    \bar d_{IP} - \bar d_{IP} \bar d_{PI}/\Delta).
  \end{align*}
  Finally, by combining this with Claim~\ref{cl:component-edges} and
  Claim~\ref{cl:component-isolated-edges} (and some tedious calculation) we
  deduce
  \begin{equation}
    \label{eq:drift:phi:2}
    \begin{aligned}
      \E\big[ \Phi(t) \bigm\vert C \subseteq I_{t-1} \big] & \leq \Phi(t - 1) -
      \frac{1}{\Delta + 1} - \frac1{10} \frac{\Delta - \bar d_{IP}}{\Delta+1} +
      \frac{1}{100} \frac{2\Delta + 1 - 2\bar d_{IP} - \bar d_{IP} \bar
      d_{PI}/\Delta}{\Delta+1} \\
      & \leq \Phi(t - 1) - \frac{2}{25} \frac{\Delta - \bar d_{IP}}{\Delta+1} -
      \frac{1}{100} \frac{ \bar d_{IP}\bar d_{PI} }{\Delta(\Delta+1)}.
    \end{aligned}
  \end{equation}

  With all these preparations we are now in a position to bound $\E[\Phi(t)]$.
  As seen in \eqref{eq:drift:phi:C3} and \eqref{eq:drift:phi:2}, both
  conditioning on components of size at least three or on vertices in isolated
  edges lead to a non-positive contribution to the drift. In order to derive an
  upper bound on $\E[\Phi(t)]$ we may thus ignore one of the terms for
  convenience. If we assume $|I_{t - 1}| \leq |M_{t - 1}|/2$ one would expect
  that the larger contribution to the change of $\Phi(t - 1)$ comes from
  components which are not isolated edges. Indeed, in that case we may ignore
  the term from \eqref{eq:drift:phi:2} and use \eqref{eq:drift:phi:C3} only to
  get
  \begin{align*}
    \E[\Phi(t)] & \stackrel{\eqref{eq:drift:phi:C3}}\le \Phi(t - 1) - \sum_{C,\;
    |V(C)| \geq 3} \frac{|V(C)|}{|V(M_{t-1})|} \frac{\bar d(C)}{25} = \Phi(t -
    1) - \frac{2|M_{t-1}\setminus I_{t-1}|}{25|V(M_{t-1})|} \\ & \le \Phi(t - 1)
    - \frac{|M_{t-1}|}{25n} \leq \Phi(t - 1) \Big(1 - \frac{1}{50 n}\Big),
  \end{align*}
  where the last inequality follows from Claim~\ref{cl:mono_phi}.

  On the other hand, suppose $|I_{t-1}| \geq |M_{t-1}|/2$ and observe that this
  implies $|V(I_{t-1})| \ge |V(M_{t-1})|/2$. This means that the probability of
  picking a vertex in $V(I_{t-1})$ to recolor is at least $1/2$ and one may hope
  that the larger contribution to the change of $\Phi(t - 1)$ comes from the
  isolated edges. Indeed, similarly as above, we now ignore the contribution
  from components of size at least three to get:
  \begin{equation}\label{eq:final}
    \E[\Phi(t)] \stackrel{\eqref{eq:drift:phi:2}}\le \Phi(t - 1) - \frac{1}{25}
    \frac{\Delta - \bar d_{IP}}{\Delta+1} - \frac{1}{200} \frac{ \bar d_{IP}\bar
    d_{PI} }{\Delta(\Delta+1)} \le \Phi(t - 1) - \frac{1}{50} \frac{\Delta -
    \bar d_{IP}}{\Delta} - \frac{1}{400} \frac{ \bar d_{IP}\bar d_{PI}
    }{\Delta^2}.
  \end{equation}

  If $\bar d_{IP} \le \Delta - \Delta\Phi(t-1)/(30n)$, then the claim follows
  just from the first term. Otherwise, by Claim~\ref{cl:mono_phi}
  \[
    \Phi(t - 1) \leq 2|M_{t - 1}| \leq 4|I_{t - 1}| \leq 2n,
  \]
  which in turn implies $\bar d_{IP} \ge \Delta(1 - \Phi(t-1)/(30n)) \geq
  14\Delta/15$. Recall, $\bar d_{PI} |P_{t - 1}| = \bar d_{IP} |V(I_{t - 1})|$,
  and note that $|V(I_{t - 1})| / |P_{t - 1}| \ge 2|I_{t - 1}| / n \geq \Phi(t
  - 1) / (2n)$. Therefore,
  \[
    \frac{1}{400} \frac{\bar d_{IP} \bar d_{PI}}{\Delta^2} = \frac{1}{400}
    \frac{|V(I_{t-1})|\cdot \bar d_{IP} \bar d_{IP}}{|P_{t-1}| \Delta^2} \geq
    \Phi(t - 1) \frac{1}{800n} \cdot \Big(\frac{14}{15}\Big)^2
  \]
  and the second term in \eqref{eq:final} is enough to conclude the proof of
  Claim~\ref{cl:mult}.
\end{proof}

As mentioned in the paragraph before Claim~\ref{cl:mult}, in order to make use
of the assertion of Claim~\ref{cl:mult}, we need a slightly different drift
theorem, one for multiplicative drift.

\begin{theorem}[Multiplicative Drift Theorem~\cite{doerr2012multiplicative}]
  \label{drift:multiplicative}
  Let $(X_t)_{t \geq 0}$ be a sequence of non-negative random variables with a
  finite state space $\mathcal{S} \subset \R^+_0$ such that $0 \in \mathcal{S}$.
  Let $s_{\mathrm{min}} := \min\{ \mathcal{S} \setminus \{0\} \}$, let $s_0 \in
  \mathcal{S} \setminus \{0\}$, and let $T := \inf \{ t \geq 0 \mid X_t = 0 \}$.
  If there exists $\delta > 0$ such that for all $s \in \mathcal{S} \setminus \{
  0 \}$ and for all $t > 0$,the case of
  \[
    \E[X_t - X_{t-1} \mid X_{t-1} = s] \leq - \delta s,
  \]
  then
  \[
    \E[T \mid X_0 = s_0] \leq \frac{1 + \ln (s_0 / s_{\mathrm{min}})}{\delta}.
  \]
\end{theorem}

Now we are ready to put things together to prove Theorem \ref{thm:main}.

\begin{proof}[Proof of Theorem~\ref{thm:main}]
  For every $t \geq 0$, we define
  \[
    \Phi'(t) =
    \begin{cases}
      \Phi(t), & \text{ if } \Phi(t) \geq n/\Delta, \\
      0, & \text{ otherwise}.
    \end{cases}
  \]
  Note that, as long as $\Phi(t-1) \geq n/\Delta$, we have $\Phi'(t-1) =
  \Phi(t-1)$ and $\Phi'(t) \leq \Phi(t)$, so the deduced bound on $\Phi(t)$ in
  Claim \ref{cl:mult} is also a bound for $\Phi'(t)$. Using Theorem
  \ref{drift:multiplicative} with Claim~\ref{cl:mult} for $T' := \inf \{ t \geq
  0 \mid \Phi'(t) = 0 \} = \inf \{ t \geq 0 \mid \Phi(t) < n/\Delta \}$, we
    get for all $s_0 >0$
  \[
    \E[T' \mid \Phi'(0) = s_0] \leq \frac{1 +
    \ln\left(\frac{s_0}{n/\Delta}\right)}{(1000n)^{-1}}.
  \]
  By Claim~\ref{cl:mono_phi} we have $\Phi'(0) \leq 2|M_0| \leq n\Delta$, and
  therefore
  \[
    \E[T'] \leq 1000n \big(1 + 2\ln{\Delta}\big) = O(n\log\Delta).
  \]
  Finally, as by Claim~\ref{cl:mono_phi} we then have $|M_{T'}| = O(n/\Delta)$,
  we conclude from Proposition~\ref{prop:endgame} that the expected number of
  steps after $T'$ to reach a legal coloring is $O(n)$. Therefore, the total
  number of required steps to reach a legal coloring is $O(n\log\Delta)$, which
  finishes the proof of Theorem~\ref{thm:main}.
\end{proof}

\section{Tail Bounds}
In this section, we prove that the runtime of the decentralized coloring algorithm is  of order $O(n \log \Delta)$ not only in expectation, but also with high probability. It turns out that this does not require much additional work, as the drift theorems are accompanied with suitable tail bounds. In many situations, concentration bounds require conditions beyond the drift, for example bounds on the step size. Notably, for multiplicative drift such additional conditions are not necessary, as the following theorem holds.
\begin{theorem}[Multiplicative Drift Tail Bound~\cite{doerr2013adaptive}]
  \label{tail:multiplicative}
  Let $(X_t)_{t \geq 0}$ be a sequence of non-negative random variables with a
  finite state space $\mathcal{S} \subset \R^+_0$ such that $0 \in \mathcal{S}$.
  Let $s_{\mathrm{min}} := \min\{ \mathcal{S} \setminus \{0\} \}$, let $s_0 \in
  \mathcal{S} \setminus \{0\}$, and let $T := \inf \{ t \geq 0 \mid X_t = 0 \}$.
  Suppose that $X_0 = s_0$, and  that there exists $\delta > 0$ such that for all $s \in \mathcal{S} \setminus \{
  0 \}$ and for all $t > 0$,
  \[
    \E[X_t - X_{t-1} \mid X_{t-1} = s] \leq - \delta s.
  \]
  Then, for all $r \geq 0$
  \[
    \Pr\Big[ T > \ceil*{\tfrac{r + \ln(s_0/s_{\min})}{\delta}} \Big] < e^{-r}.
  \]
\end{theorem}
The following proposition is a straightforward application of this theorem.
\begin{proposition}\label{prop:largeD}
Suppose $\Delta = \Omega(n^c)$ for some constant $c >0$. Then, the decentralized coloring algorithm terminates after $O(n\log \Delta)$ steps with high probability.
\end{proposition}
\begin{proof}
By Claim \ref{cl:mono_phi} we know that $s_0 \leq n \Delta$. In the proof of Theorem \ref{thm:main} we analyzed the process with multiplicative drift until the potential $\Phi$ hits $n/\Delta$. Here, we track this potential until the end of the algorithm. Note that the smallest nonzero value $\Phi$ can attain is at least $1$. Thus, under the assumption $\Delta = \Omega(n^c)$, we also have $\ln (s_0/s_{\min}) \leq C\log \Delta$ for large enough $n$, even if we do not truncate $\Phi$. Here, $C>0$ is a suitable constant, for example $C= 2(1+1/c)$. Setting $r = \log \Delta$, we can apply Theorem \ref{tail:multiplicative} to get
\[
  \Pr\big[ T > \ceil*{(1+C)1000n\log\Delta} \big] \leq e^{-\log\Delta} = O(n^{-c})= o(1),
\]
where we used $\delta = (1000n)^{-1}$ as before, due to Claim~\ref{cl:mult}.
\end{proof}

The proof of Proposition \ref{prop:largeD} fails when $\log n = \omega( \log \Delta)$. In the proof of Theorem \ref{thm:main} we switched to additive drift to analyze the second phase of the process. We use this approach again. The following tail bound for additive drift will be useful. It is a rather straightforward consequence of Azuma's inequality. Note that there is an additional assumption, namely that we have bounded step size.
\begin{theorem}[Additive Drift Tail Bound~\cite{kotzing2016concentration}]
\label{tail:additive}
Let $(X_t)_{t \geq 0}$ be a sequence of non-negative random variables with a
  finite state space $\mathcal{S} \subset \R^+_0$ such that $0 \in \mathcal{S}$.
  Let $T := \inf \{ t \geq 0 \mid X_t = 0 \}$. Suppose there are $c,\delta > 0$ such
  that for all $s \in \mathcal{S} \setminus \{ 0 \}$ and for all $t > 0$, we have both
$\E[X_{t} - X_{t-1} \mid X_t = s] \leq - \delta$ and $|X_{t+1}-X_t| < c$.
 Then, for all $r \geq 2X_0/\delta$,
  \[
    \Pr[T \geq r ] \leq \exp\Big(-\frac{r\delta^2}{8c^2}\Big).
  \]
\end{theorem}
The smaller $\Delta$ is, the less the potential changes at each step. Using this fact, the theorem above allows us to prove the next proposition. It gives us that the runtime of the algorithm is $O(n\log \Delta)$ for smaller $\Delta$.
\begin{proposition}\label{prop:smallD}
If $\Delta = O(n^{1/4})$, the decentralized coloring algorithm terminates after $O(n\log \Delta)$ steps with high probability
\end{proposition}
\begin{proof}
We go back to splitting the process in two phases as in the proof of Theorem
\ref{thm:main}. Let $T_1$ and $T_2$ be the duration of Phase 1 and 2
respectively. We consider Phase 1 first. To be able to apply Theorem
\ref{tail:additive}, we use the potential function $\Psi(t) :=
\max\{\log(\Delta\Phi(t)/n),0\}$. As we will see, the logarithm converts multiplicative drift into additive drift. Note that $T_1 = \inf \{ t \geq 0 \mid \Psi(t) = 0 \}$. Using Jensen's inequality and Claim \ref{cl:mult}, we get for all $s \in S\setminus \{0\}$ and $t \geq 0$
\begin{align*}
\E[\Psi(t+1) - \Psi(t) \mid \Psi(t) =s] &= \E\left[\log\left(\frac{\Delta\Phi(t+1)}{n}\right) - \log\left(\frac{\Delta\Phi(t)}{n}\right) \middle\vert \Psi(t) = s\right] \\
 &= \E\left[\log\left(\frac{\Phi(t+1)}{\Phi(t)}\right) \middle\vert \Psi(t) = s\right] \\
 &\leq \log \E\left[\frac{\Phi(t+1)}{\Phi(t)} \middle\vert \Psi(t) = s\right] \\
 &\leq \log\left(1 - \frac{1}{1000n}\right) \\
 &\leq - \frac{1}{1000n}.
\end{align*}
The last inequality follows as $\log x \leq x-1$ for all $x > 0$. Now that we have determined the drift, we need to bound the step size. $|M_t|$ and $|I_t|$ can change by at most $\Delta$ at each step and $e(V(I_t,P_t)$ by at most $\Delta^2$. Thus, the step size of $\Phi$ is bounded by $2\Delta$. As the logarithm is concave, the largest effect of such a change in $\Phi$ on the value of $\Psi$ is when $\Phi$ is as small as possible. In particular, this is the case if $\Phi$ goes from $2\Delta + n/\Delta$ to $n/\Delta$. Thus we have
\[
|\Psi(t+1)- \Psi(t)| \leq \left|\log\left(\frac{2\Delta+\frac{n}{\Delta}}{\frac{n}{\Delta}}\right)\right| = \log\left(1 + \frac{2\Delta^2}{n}\right) < \frac{2\Delta^2}{n}.
\]
Therefore, $2\Delta^2/n$ is a bound on the step size of $\Psi$. As $\Phi(0) \leq n\Delta$ we have $\Psi(0) \leq 2\log \Delta$. Hence we can use Theorem \ref{tail:additive} with $r= 4000n(1+\log \Delta)$. We get
\[
  \Pr[T_1 \geq r] \leq \exp\left(-\frac{4000n(1+\log \Delta) \cdot
  \left(\frac{1}{1000n}\right)^2}{8\left(\frac{2\Delta^2}{n}\right)^2}\right) =
  \exp\left(-\frac{(1+\log \Delta)n}{8000\Delta^4}\right) =o(1).
\]
We turn our attention to Phase 2. Here, we already have additive drift for $|M_t|$, so we can use Theorem \ref{tail:additive} immediately. Claim \ref{cl:component-edges} gives us
\[
  \E\big[ |M_t| - |M_{t-1}| \bigm\vert |M_{t-1}| = s \big] \leq
  -\frac{1}{\Delta+1} \leq -\frac{1}{2\Delta},
\]
for all $s ,t > 0$. We also have $\big||M_t| - |M_{t-1}|\big| \leq \Delta$ for all $t> 0$. By Claim \ref{cl:mono_phi} we get $|M_t| \leq 2n/\Delta$ at the start of Phase 2. Hence we can apply Theorem \ref{tail:additive} with $r= (4 + \log \Delta)n$. We get
\[
  \Pr[T_2 \geq r] \leq \exp\left(-\frac{(4+\log \Delta)n \cdot
  \left(\frac{1}{2\Delta}\right)^2}{8\Delta^2}\right) = \exp\left(-\frac{(4 +
  \log\Delta) n}{32\Delta^4}\right) = o(1).
\]
Using a union bound gives us that the algorithm terminates in $O(n \log \Delta)$ round with high probability.
\end{proof}
Proposition \ref{prop:largeD} and \ref{prop:smallD} cover all possible values of
$\Delta$. Therefore, the algorithm has runtime $O(n  \log \Delta)$ with high
probability for any $\Delta$.

\section{A simultaneous-recoloring variant of the algorithm}

A natural question is whether the original algorithm can be parallelized. So
what if instead of choosing one conflicted vertex at a time in
Step~\ref{alg-vertex}\ {\em all} conflicted vertices would simultaneously want
to change their color? It turns out that this process does not even have
polynomial runtime on the complete graph $K_n$.

\begin{proposition}
  The Decentralized Graph Colouring algorithm in which all conflicted vertices
  choose a new color uniformly at random needs $e^{\Omega(n)}$ rounds in
  expectation to terminate on a complete graph on $n$ vertices $K_n$.
\end{proposition}
\begin{proof}
  Fix a sufficiently small constant $\eps >0$, e.g. $\eps = 0.1$. For a round
  $t$, let $X_t$ be the number of conflicted vertices, i.e., the number of
  vertices whose color is not unique. Due to symmetry, $X_t$ is a Markov
  process. Let $T_{\eps}$ be the first round in which $X_t \le \eps n$. We
  show that $T_\eps$ has exponentially large expectation. Consider any round $t$
  with $X_t = x > \eps n$. Then we show that
  \begin{equation}\label{eq:exp_bound}
    \Pr[X_{t+1} \le \eps n \mid X_t = x] = e^{-\Omega(n)},
  \end{equation}
  where the hidden constant is independent of $x$. We remark that the same
  argument also shows that with high probability $X_1 >\eps n$, since the
  initial round is formally equivalent to the hypothetical case $X_0 = n$. So
  the proposition follows if we can show~\eqref{eq:exp_bound}.

  To show \eqref{eq:exp_bound}, we uncover the new colors in two batches. In the
  first batch, we uncover the colors of all but $\eps n$ vertices. If there are
  more than $\eps n$ vertices in conflicts from the first batch, then there is
  nothing to show. So in the following we may assume (and implicitly condition)
  on the opposite event that uncovering the first batch creates at most $\eps n$
  conflicted vertices. This implies that the set $C_1$ of colors appearing among
  the $(1-\eps) n$ uncovered vertices, has size at least $|C_1| \ge (1-2\eps)n$.
  Let $C_2 \subseteq C_1$ be the set of colors in $C_1$ that also appear in the
  second batch, i.e., for which a conflict is created by the second batch. The
  probability that a fixed color in $C_1$ does \emph{not} occur in $C_2$ is
  $(1-1/n)^{\eps n} = e^{-\eps + O(1/n)} \leq 1-7\eps/8$ for sufficiently large
  $n$, where we use that $e^{-\eps} < 1-7\eps/8$ for $\eps < 0.2$. Hence,
  $\E[|C_2|] \geq 7\eps/8\cdot (1-2\eps)n \geq 5\eps/8\cdot n$ for $\eps \le
  1/7$.

  The size of $C_2$ is given by the number of non-empty bins in a {\em
  Balls-and-Bins} problem, and this number is known to be concentrated around
  its expectation since the number of empty bins is negatively associated, and
  thus the Chernoff bounds are applicable. Since this is a well-known argument,
  we refrain from spelling out the details and refer the reader to the standard
  exposition~\cite[Proposition 29 and Section 3.3]{dubhashi1996balls}. The
  result is that $\Pr[|C_2| \leq \eps/2\cdot n] = e^{-\Omega(n)}$.

  It remains to observe that $X_{t+1} > 2|C_2|$, since every color in $C_2$
  causes at least two conflicted vertices (one from the second batch and one
  from the rest). Hence, $\Pr[X_{t+1} \leq \eps n] \le \Pr[|C_2| \leq
  \eps/2\cdot n] = e^{-\Omega(n)}$, as required.
\end{proof}

\bibliography{references}

\end{document}